\documentclass[prl,aps,amssymb,shownopacs,twocolumn,superscriptaddress,times,10pt]{revtex4}
\usepackage{amsmath}
\usepackage{amssymb}
\usepackage{amsthm}
\usepackage{amsfonts}
\usepackage{listings}
\usepackage{enumerate}
\usepackage{latexsym}
\usepackage{psfrag}
\usepackage{bm}
\usepackage[all]{xy}
\usepackage{graphicx}
\usepackage{subfigure}
\usepackage[pdftex,colorlinks=false]{hyperref}
\usepackage{xcolor}
\usepackage{morefloats}%needed to make marginpar behave

\usepackage{times}

\newcommand{\beq}{\begin{equation}}
\newcommand{\eneq}{\end{equation}}
\def\be{\begin{equation}}
\def\ee{\end{equation}}
\def\ba{\begin{eqnarray}}
\def\ea{\end{eqnarray}}

\def\R{{\rm Re}}
\def\Z{\mathbb{Z}}
\def\C{\mathbb{C}}

\newcommand{\eqnref}[1]{Eq.~\eqref{#1}}
\newcommand{\punc}[1]{\,#1}

\def\mT{\mathcal{T}}

\def\beq{\begin{equation}}
\def\eeq{\end{equation}}
\def\barray{\begin{eqnarray}}
\def\earray{\end{eqnarray}}

\font\upright=cmu10 scaled\magstep1
\def\stroke{\vrule height8pt width0.4pt depth-0.1pt}

\def\Zmath{\mathbb{Z}}
\def\Qmath{\vcenter{\hbox{\upright\rlap{\rlap{Q}\kern
                   3.8pt\stroke}\phantom{Q}}}}
\def\Nmath{\vcenter{\hbox{\upright\rlap{I}\kern 1.7pt N}}}
\def\Cmath{\vcenter{\hbox{\upright\rlap{\rlap{C}\kern
                   3.8pt\stroke}\phantom{C}}}}
\def\Rmath{\vcenter{\hbox{\upright\rlap{I}\kern 1.7pt R}}}
\def\Z{\ifmmode\Zmath\else$\Zmath$\fi}
\def\Q{\ifmmode\Qmath\else$\Qmath$\fi}
\def\N{\ifmmode\Nmath\else$\Nmath$\fi}
\def\C{\ifmmode\Cmath\else$\Cmath$\fi}
\def\R{\ifmmode\Rmath\else$\Rmath$\fi}

\input{epsf}

\newtheorem{mylem}{Lemma}

\newcounter{defcounter}
\begin{document}

\tolerance 10000

\newcommand{\cbl}[1]{\color{blue} #1 \color{black}}

\newcommand{\vk}{{\bf k}}

\title{No-Go Theorem for Boson Condensation in Topologically Ordered Quantum Liquids}

\author{
Titus Neupert}
\address{
 Department of Physics, University of Zurich, Winterthurerstrasse 190, 8057 Zurich, Switzerland
}
\address{
 Princeton Center for Theoretical Science, Princeton University, Princeton, New Jersey 08544, USA
}

\author{
Huan He}
\address{
Department of Physics, Princeton University, Princeton, New Jersey 08544, USA
}

\author{
Curt von Keyserlingk}
\address{
 Princeton Center for Theoretical Science, Princeton University, Princeton, New Jersey 08544, USA
}

\author{
Germ\'an Sierra}
\address{
Instituto de F\'isica Te\'orica, UAM-CSIC, Madrid, Spain
}

\author{
B. Andrei Bernevig}
\address{
Department of Physics, Princeton University, Princeton, New Jersey 08544, USA
}

\begin{abstract}
Certain phase transitions between topological quantum field theories (TQFT) are driven by the condensation of bosonic anyons. However, as bosons in a TQFT are themselves nontrivial collective excitations, there can be topological obstructions that prevent them from condensing. Here we formulate such an obstruction in the form of a no-go theorem. We use it to show that no condensation is possible in SO(3)$_k$ TQFTs with odd $k$. We further show that a ``layered'' theory obtained by tensoring SO(3)$_k$ TQFT with itself any integer number of times does not admit condensation transitions either. This includes (as the case $k=3$) the noncondensability of any number of layers of the Fibonacci TQFT.
%If a TQFT does not contain bosons, one can tensor it with itself to obtain a ``multi-layer'' theory that will contain bosons for a sufficient number of layers. Condensing these bosons can lead to a $\mathbb{Z}_m$ grading of the TQFT under layering. An important example is the $\mathbb{Z}_{16}$ grading of the Ising TQFT (physically equivalent to 16 types of gauged chiral superconductors). It is known from an example that not all TQFTs admit such a $\mathbb{Z}_m$ grading under layering, that is, the bosons in the multi-layer theories cannot condense. Here, we give a general sufficient (not necessary) set of conditions for the noncondensability of bosons in multilayer TQFTs. As an example, we show that the SO(3)$_k$ TQFTs with $k$ odd, including the case of the Fibonacci TQFT, do not admit condensation transitions for any number of layers.  
\end{abstract}

\date{\today}
\pacs{03.67.Mn, 05.30.Pr, 73.43.-f}

\maketitle

%\section{Introduction}

Topological order, a fundamental concept in quantum many-body physics, is best understood in two-dimensional gapped quantum liquids, such as the fractional quantum Hall effect and certain spin liquids~\cite{FQH1,FQH2,Kitaev06,Kitaev,WWZ89,Wen89,Wen90,CGW10,Wen15}. In these systems, quasiparticle excitations with anyonic quantum-statistical properties emerge~\cite{W82}. Their fusion and braiding behavior at large distances define a topological quantum field theory (TQFT), which characterizes the universal properties of the phase~\cite{Bais3,HeMoradiWen,MoradiWen,LiuWangYouWen}.

The phase transitions between topological phases are, most of the times, driven by the condensation of bosons~\cite{KO01,Bais1,Bais2,Bais3,YJW12,BW10,BSS12,BSS11,HW13,L14,HW15b,Neupert16}. 
In the context of TQFTs, a boson is an emergent quasiparticle in the topologically ordered phase with bosonic self-statistics, but which could have nontrivial fusion and braiding relations with the other anyons. Such a quasiparticle can potentially undergo Bose-Einstein condensation, causing a phase transition to another topologically ordered phase. %. However, for the theory to remain physically consistent, the condensation has to be accompanied by a transition to a different from of topological order. 
%Universal properties 
The topological data of the new phase can be inferred from those of the initial topological order \cite{Neupert16}. 

One motivation to study condensation transitions is to classify topological order. An important example are the 16 types of gauged chiral superconductors introduced by Kitaev~\cite{Kitaev06}. Kiteav showed that while two-dimensional superconductors are classified by an integer $\mathbb{Z}$, only $16$ bulk phases are topologically distinct. This construction can be understood by considering several layers of the elementary (Ising) TQFT. Coupling the layers by condensing inter-layer cooper pairs, one obtains exactly 16 distinct TQFTs including Ising, the toric code and the double semion model. They determine the nature of the topologically protected excitations in the vortices of each superconductor, including their braiding statistics. In essence, $\mathbb{Z}_{16}$ classification is a property of the Ising TQFT. 

It is imperative to ask whether multi-layer systems of other TQFTs show a similar collapse of the classification from $\mathbb{Z}$ to $\mathbb{Z}_{N}$ for some integer $N$. In this paper, we derive a criterion for when this is not the case, i.e., when the $\mathbb{Z}$ classification generated by a given TQFT is stable. This criterion is based on the fact that there exist
 bosonic anyons that cannot be condensed. An example are the bosons in multi-layered Fibonacci topological order~\cite{Bookera12,Bais2,Neupert16}.
In this work, we generalize this observation by formulating a no-go theorem that constitutes a sufficient obstruction against the condensation of a boson.
Our criterion and its proof are given using the tensor category formulation of topological order~\cite{WittenJonesPoly, FFRC06, LW05, LW14, KW14, KWZ15,Kitaev06,Bonderson07,BernevigNeupert15}, which we can use to describe the condensation transition axiomatically~\cite{Bais1,Bais2,Neupert16}. We apply our no-go theorem to several examples, including the forementioned multi-layer Fibonacci TQFTs.

%\section{Formalism of anyon condensation}
{\it Formalism} ---
We use the algebraic formulation of anyon condensation discussed in Ref.~\onlinecite{Neupert16}. Here we simply restate the important relations and refer the reader to Ref.~\onlinecite{Neupert16} for details.
A \emph{fusion category} is characterized by a set of anyons $a,b,c,\ldots$ and fusion rules $a\times b=\sum_c N^c_{ab} c$ between them. The quantum dimension $d_a$ gives the size of the nonlocal internal Hilbert space associated with anyon $a$, and is equal to largest eigenvalue of the matrix $N_a$ with elements $(N_a)_{bc}\equiv N^c_{ab}$. A \emph{braided tensor category} has additional structure, of which we will use the topological spin $\theta_a$ of $a$, a complex number with $|\theta_a|=1$. 
Bosons are defined by $\theta_a=1$.
A special role is played by the vacuum anyon as the unique identity element of fusion. It is a boson with quantum dimension 1. 

Condensation is based on a mapping, called restriction, between the anyons $a$ in the original TQFT $\mathcal{A}$ and the anyons $t$ in the condensed fusion category $\mathcal{T}$ characterized by integers $n^t_a\in \mathbb{Z}_{\geq0}$:
\begin{equation}
a\mapsto a^{\downarrow} \equiv \sum_{t\in \mT} n_a^t t,\quad\forall a\in\mathcal{A}.
\label{eq: restriction map}
\end{equation}
If more than one particle appears on the right-hand side of Eq.~\eqref{eq: restriction map}, we say that the $a$ particle \emph{splits}. If $n^t_a\neq0$, we say $t$ is in the \emph{restriction} of $a$ or $t \in a^{\downarrow}$.
We require that $n^\varphi_1=1$, where $\varphi$ and $1$ are the vacua in $\mathcal{T}$ and $\mathcal{A}$, respectively. 
Imposing that condensation commutes with fusion implies the fundamental relation~\cite{Neupert16}
\begin{eqnarray}
\label{eq.commuteRstrFus}
\sum_{r,s\in\mT} n^r_a n^s_b \tilde{N}^t_{rs}=\sum_{c\in\mathcal{A}} N^c_{ab} n^t_c,
\end{eqnarray}
between the fusion coefficients $N^c_{ab}$ in $\mathcal{A}$ and the fusion coefficients $\tilde{N}^t_{rs}$ in ${\cal T}$. A corollary to Eq.~\eqref{eq.commuteRstrFus} \cite{Neupert16} is 
%that the quantum dimensions satisfy
\begin{equation}
d_a=\sum_{r\in\mathcal{T}}n^r_a d_r,\qquad \forall \ a\ \in \mathcal{A}.
\label{eq: quantum dimensions}
\end{equation}
The restriction is compatible with conjugation to antiparticles, i.e.,
$n^t_a=n^{\bar{t}}_{\bar{a}}$,
where bar denotes the (unique) antiparticle of an anyon. 
We say particle $a$ \emph{condenses} if $\varphi \in a^{\downarrow}$,
i.e., $n^\varphi_a\neq0$. Common knowledge in condensed matter physics says that any bosons can condense.
However, it may also occur that a specific boson $a$ cannot condense, i.e., there is no solution to the above equations with $n^\varphi_a\neq0$. 
This is the situation we shall analyze in this paper.

%\section{First No-Go Theorem}
%\label{sec: no-go thm 1}

Finally, the following definition is useful for formulating our no-go theorem: For a given anyon $b$, a subset $\mathcal{I}_b=\{a_1,\ldots,a_m\}$ of anyons is called a \emph{set of zero modes 
localized by $b$}~\cite{Barkeshli14} if for all $i,j=1,\ldots, m$:
\begin{enumerate}
\item The fusion products $a_i\times a_j$ do not contain condensable bosons, except the identity if $a_i=\bar{a}_j$,~\cite{footnote-Ib-set}
\item all $a_i$ are zero modes of $b$, by which we mean $a_i\times b= b+\ldots$, (i.e. $N_{a_i b}^b > 0$)
\item if a particle $a_i$ is in $\mathcal{I}_b$ then so is its antiparticle.
\end{enumerate}
Note that the choice of $\mathcal{I}_b$ for a given boson $b$ is not unique and that $\mathcal{I}_b$ may or may not contain the identity. (The above conditions are satisfied in both cases.) Typically, we will be interest to find a set $\mathcal{I}_b$ that is as large as possible. To motivate the terminology of the set $\mathcal{I}_b$, observe that $N_{a b}^b>0$ implies that $a$ anyons can always be emitted or absorbed by $b$. Therefore, $b$ must carry a zero-mode excitation of $a$.
We can now state our first main result, a general condition under which a boson $B$ cannot condense. It is an obstruction that is sufficient to show that condensation of $B$ cannot occur.

{\it No-go theorem} ---
A boson $B$ cannot condense if there exists a set $\mathcal{I}_B$, such that the sum of the quantum dimensions of all anyons in $\mathcal{I}_B$ exceeds the quantum dimension of $B$, i.e., if 
\begin{equation}
d_B<d_{a_1}+d_{a_2}+\cdots+d_{a_m}\punc{.}
%\sum_{i=1}^m d_{a_i}.
\label{eq: contradiction}
\end{equation}

%{\it Proof} ---

\begin{figure}[t]
\begin{center}
\includegraphics[width=0.27 \textwidth]{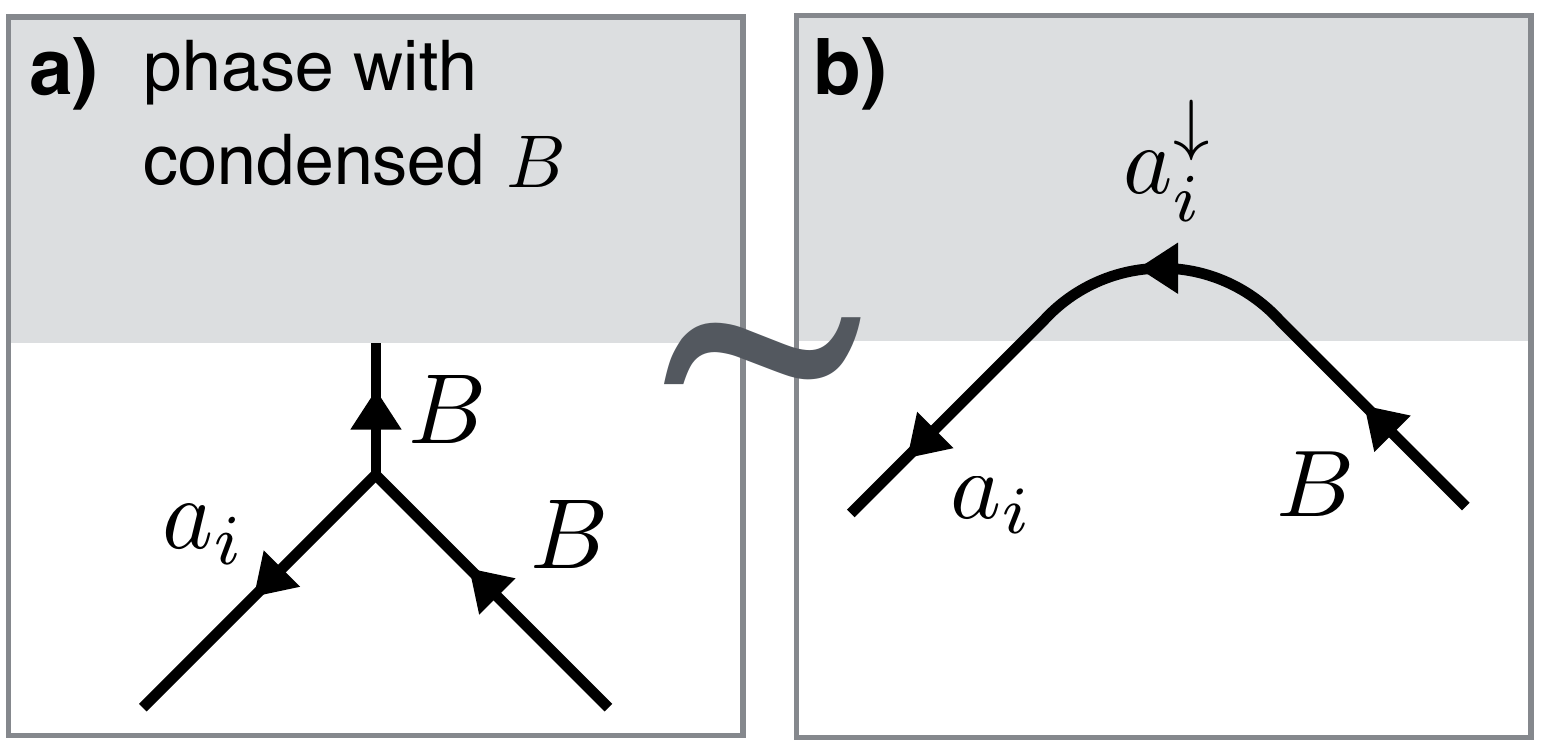}
\caption{
Tunneling processes mediated by an anyon condensate. 
The gray region is a phase in which a boson $B$ is condensed. 
a) Vertex of a boson $B$ that localizes a zero mode of anyon $a_i$. 
In the condensed phase, $B$ can be converted into an identity particle world line (not shown).  
By the axioms of anyon condensation, processes a) and b) are equivalent, i.e., $B$ can be converted into $a_i$ by tunneling through the condensate. 
}
\label{fig: physical picture}
\end{center}
\end{figure}

\begin{proof}
We start by showing that all particles in $\mathcal{I}_B$ do not split, and have distinct restrictions. This follows from inspection of Eq.~\eqref{eq.commuteRstrFus} for $t=\varphi$, $a=a_i$, $b=\bar{a}_j$,
\begin{equation}
\sum_{r\in \mathcal{T}}n^r_{a_i}n^r_{a_j}
=\delta_{i,j}
+
\sum_{c\neq 1} N_{a_i \bar{a}_j}^c n_c^\varphi,
\end{equation} 
where we used $n^\varphi_c=n^\varphi_{\bar{c}}$.
By assumption, there are no condensable bosons in $a_i\times \bar{a}_j$, hence $N_{a_i \bar{a}_j}^c$ and $n^\varphi_c$ cannot be both nonzero for any $c\neq 1$. Thus $\sum_r n^r_{a_i}n^r_{a_i}=1$, implying a single restriction $a^\downarrow_i$ of $a_i$, with $d_{a^{\downarrow}_i}=d_{a_i}$ using \eqnref{eq: quantum dimensions}. Moreover, $\sum_r n^r_{a_i}n^r_{a_j}=0$ if $i\neq j$, implying that the restrictions of $a_i\neq a_j$ are distinct particles.

With this knowledge about the restrictions of the $a_i$, Eq.~\eqref{eq.commuteRstrFus} for $t=\varphi$, $a=a_i$, $b=\bar{B}$ evaluates to
\begin{equation}
n^{\bar{a}^\downarrow_i}_{\bar{B}}=n^{a^\downarrow_i}_B
=\sum_c N^{\bar{c}}_{a_i \bar{B}} n^\varphi_c
\geq
N^B_{a_i B} n^\varphi_B,
\end{equation}
where we used $N^{\bar{B}}_{a_i \bar{B}}=N^B_{a_i B}$
Inserting this inequality in Eq.~\eqref{eq: quantum dimensions} for $a=B$, and using $d_{a_i}=d_{a^\downarrow_i}$, we have
\begin{equation}
d_B\geq n^\varphi_B\sum_{i=1}^m N^B_{a_i B}d_{a_i}.
\label{eq: quantum dimension condition B}
\end{equation}
It follows that in a situation where Eq.~\eqref{eq: contradiction}
 holds, Eq.~\eqref{eq: quantum dimension condition B} implies $n^\varphi_B=0$, i.e., $B$ does not condense. 
[Note that in the case $N^B_{a_i B}>1$, a stronger form of Eq.~\eqref{eq: contradiction} with $d_{a_i}$ is replaced by $N^B_{a_i B}d_{a_i}$ holds.]

To follow up with a pictorial representation of these equations, consider the tunneling of anyons across the domain wall as shown in Fig.~\ref{fig: physical picture}, where each particle $a$ in the uncondensed theory is converted into its restriction $a^{\downarrow}$ in the gray region.  Figure~\ref{fig: physical picture}~(a) shows a vertex allowed by the fusion rule $a_i\times B\to B$ in the uncondensed phase. The boson $B$ enters the condensed phase, where it can disappear as it is part of the condensate (one of its restrictions is the vacuum $\varphi$, the world lines of which can be removed at will). By the fundamental assumption that fusion and condensation commute [which is at the heart of Eq.~\eqref{eq.commuteRstrFus}], Fig.~\ref{fig: physical picture}~(a) is equivalent to Fig.~\ref{fig: physical picture}~(b). The latter represents a coherent tunneling process that is mediated by the condensate and converts $B$ into any of the $a_i$. The existence of this process implies that the distinct restriction $a^\downarrow_i$ of any $a_i$ must be in the restriction of $B$. 
Hence, by Eq.~\eqref{eq: quantum dimensions}, the quantum dimension of $B$ must be large enough to accommodate all the distinct restrictions of the $a_i$, if $B$ condenses. Therefore if we find sufficiently many $a_i$ such that Eq.~\eqref{eq: contradiction} holds, $B$ cannot condense. 
\end{proof}
Note that the no-go theorem does not a priori require knowing the braiding data of $\mathcal{A}$ -- although the modular tensor category structure fixes that data to some extend. The theorem involves only data obtainable from $N_{ab}^c$.
We remark that the no-go theorem can only ever yield an obstruction against the condensation of non-Abelian bosons. For Abelian bosons, the theory after condensation can be constructed explicitly, which is a constructive proof that there is no obstruction.~\cite{Neupert16}

We now demonstrate that the no-go theorem is practically useful by considering three examples: (i) multiple layers of the Fibonacci TQFT, (ii) single layers of the SO(3)$_k$ TQFT for $k$ odd, and (iii) multiple layers of the latter. We will show that all these theories, while containing bosons, do not admit condensation transitions. All the bosons are noncondensable. 
Additional general results, concerning for instance TQFTs with a condensing Abelian sector and with only a single boson, are given in the Supplemental Information.~\cite{Supp}

{\it Example (i): Multiple layers of Fibonacci} ---
The Fibonacci category $\mathcal{A}_{\mathrm{Fib}}$ is a non-Abelian TQFT containing just one nontrivial particle $\tau$ with a fusion rule
$\tau\times\tau=1+\tau$, a topological spin $\theta_\tau=e^{\mathrm{i}4\pi/5}$, and a quantum dimension $d_\tau=\phi$  given by the golden ratio $\phi=(1+\sqrt{5})/2$. As $\mathcal{A}_{\mathrm{Fib}}$ does not contain any nontrivial boson, it cannot undergo a condensation transition. We are interested whether the TQFT formed by $N$ identical layers of $\mathcal{A}_{\mathrm{Fib}}$ i.e., the TQFT $\mathcal{A}_{\mathrm{Fib}}^{\otimes N}$, admits a condensation transition. 
The TQFT $\mathcal{A}_{\mathrm{Fib}}^{\otimes N}$ contains $2^N$ particles corresponding to all possible distributions of $\tau$-particles over the $N$ layers.  For each $r=0,\ldots,N$ there are $\binom{N}{r}$ so-called  $(r\tau)$ particles with $\tau$'s in exactly $r$ layers, each with spin $\theta_{(r\tau)}=e^{\mathrm{i}4\pi r/5}$ and quantum dimension $d_{(r\tau)}=\phi^r$. The unique $r=0$ particle is the identity of $\mathcal{A}_{\mathrm{Fib}}^{\otimes N}$. From the topological spin, the bosons in $\mathcal{A}_{\mathrm{Fib}}^{\otimes N}$ are $(r\tau)$ particles with $r=5n, \ n\in \mathbb{Z}$. Using the no-go theorem, we show that none of these bosons can condense. 

Using proof by induction on $n\geq 1$, we show that for any $(5n\tau)$ boson $B$, there exists a set $\mathcal{I}_{(5n\tau)}$ such that Eq.~\eqref{eq: contradiction} holds. We first consider the case $n=1$. Given a $(5\tau)$ boson, we must construct a set $\mathcal{I}_{(5\tau)}$ for this boson. Consider the set formed by all $(2\tau)$ particles obtained by replacing any $3$ $\tau$'s in the boson with a 1. There are $\binom{5}{2}=10$ such  $(2\tau)$ particles for a given $(5\tau)$ boson. They form a set $\mathcal{I}_{(5\tau)}$ that obeys point 1--3 from the definition: point 1 holds as any product of two of these particles has at most 4 $\tau$s and is therefore not a (potentially condensable) boson. Points 2 and 3 can be checked by using the Fibonacci fusion rules in each layer. Finally, Eq.~\eqref{eq: contradiction} holds because
\begin{equation}
d_{(5\tau)}=\phi^5<10 \phi^2=\sum_{a_i\in \mathcal{I}_{(5\tau)}}d_{a_i}
\end{equation}
evaluates to about $11.1<26.2$. We conclude that none of the $(5\tau)$ bosons condense for any number $N$ of layers of Fibonacci TQFT.

For the induction step, we assume that none of the $(5n\tau)$ bosons can condense for $n<n_0$, $n_0>1$, and we show that the same holds for the $(5n_0\tau)$ bosons. Define $r_0:=\lfloor (5n_0-1)/2\rfloor$, where $\lfloor x\rfloor$ is the largest integer smaller than or equal to $x$. For a given $(5n_0\tau)$ boson, form the set $\mathcal{I}_{(5n_0\tau)}$ out of all $(r_0\tau)$-particles that are obtained by replacing any $(5n_0-r_0)$ $\tau$'s in the boson $(5n_0\tau)$ with a 1. There are $\binom{5n_0}{r_0}$ such  $(r_0\tau)$ particles. They form a set $\mathcal{I}_{(5n_0\tau)}$ for $(5n_0\tau)$. In particular their fusion products can only contain $(5n\tau)$-bosons with $n<n_0$, which cannot condense by assumption. Equation~\eqref{eq: contradiction} reads for this case
\begin{equation}
\phi^{5n_0}<\binom{5n_0}{r_0}\phi^{5n_0-r_0}. 
\label{eq: inequality Fib induction step}
\end{equation}
Using that $r_0\sim 5n_0/2$ and $\binom{5n_0}{5 n_0/2}\sim4^{5n_0/2}/\sqrt{\pi 5n_0/2}$ for large $n_0$, we obtain that the right-hand side of Eq.~\eqref{eq: inequality Fib induction step} grows like $4^{5n_0/2}\phi^{5n_0/2}/\sqrt{n_0}$, asymptotically dominating the left-hand side. An explicit evaluation yields that Eq.~\eqref{eq: inequality Fib induction step} holds for any $n_0\geq1$ in fact.
% Using the chain of inequalities
 %\begin{equation}
% \phi^{r_0}<
 %\phi^{5n_0/2}
 %\leq
 %\left(\frac{10n_0}{5n_0-2}\right)^{\frac52 n_0-1} 
 %\leq
% \left(\frac{5n_0}{r_0}\right)^{r_0}
% \leq
% \binom{5n_0}{r_0},
% \end{equation}
%one can see that Eq.~\eqref{eq: inequality Fib induction step} holds for any $n_0\geq1$. 
We have thus shown that none of the $(5n_0\tau)$ bosons can condense. This concludes the induction step and the proof that no boson in $\mathcal{A}_{\mathrm{Fib}}^{\otimes N}$ can condense. 

{\it Example (ii): Single layer of} SO(3)$_k$ ---
Our second example  focuses on the (single-layer) TQFTs associated with the Lie group SO(3) at values of odd level $k$. They contain bosons for an infinite subset of $k$. We  show that none of these bosons can condense. The SO(3)$_k$ TQFTs with $k$ odd have $(k+1)/2$ anyons 
 $j=0,\cdots, (k-1)/2$ with 
\begin{equation}
d_j=\frac{\sin\left(\pi \frac{2j+1}{k+2}\right)}{\sin\left[\pi/(k+2)\right]},\qquad
\theta_j=e^{2\pi\mathrm{i}j \frac{j+1}{k+2}}.
\label{eq: quantum dimensions SO(3)}
\end{equation}
We note that for $k$ odd, all particles have distinct quantum dimensions. 
The fusion rules are
\begin{equation}
N^{j_3}_{j_1j_2}=
\begin{cases}
1&|j_1-j_2|\leq j_3\leq \mathrm{min}\{j_1+j_2,k-j_1-j_2\}\\
0&\text{else}
\end{cases}.
\label{eq: fusion rules SO(3)}
\end{equation}
The smallest odd $k$ for which SO(3)$_k$ contains a boson is $k=13$, in which $j=5$ is a boson -- an uncondensable one, as we shall see. 

\begin{figure}[t]
\begin{center}
\includegraphics[width=0.49 \textwidth]{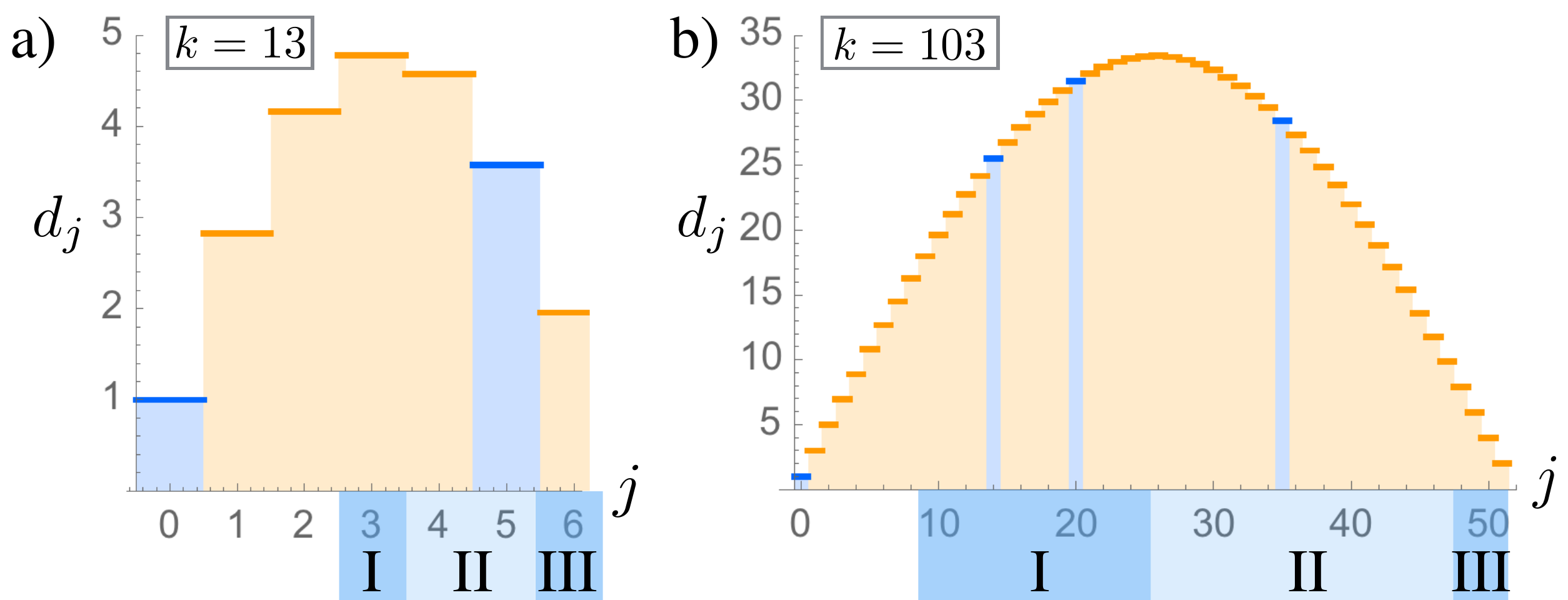}
\caption{
Quantum dimensions and bosons (blue columns) for SO(3)$_k$ theories with a) $k=13$ and b) $k=103$. These are the smallest $k$, for which SO(3)$_k$ contains two and four bosons, respectively. Indicated are also the ranges I--III defined in Eq.~\eqref{eq: ranges}. The maximum quantum dimension coincides with the boundary between range I and II in Eq.~\eqref{eq: ranges}.
For instance, to apply the no-go theorem to the $j=5$ boson in a), choose $\mathcal{I}_{j=5}=\{j=2\}$ and use that $d_5\approx3.6$ is smaller than $d_2\approx4.2$.
}
\label{fig: regions examples}
\end{center}
\end{figure}

The topological spins $\theta_j$ yield the condition
$j(j+1)=k+2$
for the lowest $j$ that may correspond to a boson (aside from the vacuum $j=0$). (Frequently, this condition cannot be met with integer $j$, as in the $k=13$ example, and the lowest boson appears at even higher $j$.)
We conclude that the first boson after $j=0$ cannot occur for $j$ lower than
\begin{equation}
j_0=\left\lfloor\sqrt{k+9/4}-1/2\right\rfloor.
\label{eq: j0}
\end{equation}
We will now discuss separately bosons $j$ in the three ranges
(see Fig.~\ref{fig: regions examples} for two examples)
\begin{subequations}
\begin{eqnarray}
\text{I.}&&\quad
j_0\leq j\leq \lfloor k/4\rfloor
\label{eq: range 1}
\\
\text{II.}&&\quad
\lfloor k/4\rfloor<j\leq\frac{k-1}{2}-\left\lfloor\frac{j_0-1}{2}\right\rfloor
\label{eq: range 2}
\\
\text{III.}&&\quad
\frac{k-1}{2}-\left\lfloor\frac{j_0-1}{2}\right\rfloor<j\leq\frac{k-1}{2}.
\label{eq: range 3}
\end{eqnarray}
\label{eq: ranges}
\end{subequations}

Due to Eq.~\eqref{eq: j0}, bosons $j_B$ in range III have no bosons in their fusion product $j_B\times j_B$, other than the identity. Thus, from Eq.~\eqref{eq.commuteRstrFus} for $t=\varphi$, and the fact that $B$ are their own antiparticles, we conclude that they cannot split. Using Eq.~\eqref{eq: quantum dimensions} and the fact that they have $d_{j_B}>1$, we conclude that they cannot restrict to the vacuum i.e., they cannot condense.

We now use our no-go theorem to show that bosons $j_B$ in range I are non-condensable. Specifically, we show that the particles $0<j<\lfloor j_B/2\rfloor$ form a set $\mathcal{I}_{j_B}$ of $j_B$ obeying Eq.~\eqref{eq: contradiction}. Before establishing that they satisfy the conditions for a set $\mathcal{I}_{j_B}$, let us show that Eq.~\eqref{eq: contradiction} holds for $\mathcal{I}_{j_B}$. For large $k$, we can rely on the following asymptotic estimate. Using that the sine function in Eq.~\eqref{eq: quantum dimensions SO(3)} is monotonously increasing with negative second derivative for $j\leq\lfloor k/4\rfloor$, the estimate 
\begin{equation}
2 j_B+1<\sum_{j=1}^{\lfloor j_B/2\rfloor -1}(2j+1)
\label{eq: range I d estimate}
\end{equation}
implies Eq.~\eqref{eq: contradiction} for $j_B$ in range I. This inequality holds for all $j_B\geq 10$. Using Eq.~\eqref{eq: j0} we conclude that it applies to all bosons in range I for $k\geq109$. We verified explicitly that inequality~\eqref{eq: contradiction} holds (using the exact values of the quantum dimensions) for all bosons in range I for $k<109$. Finally, it is readily verified using Eq.~\eqref{eq: fusion rules SO(3)} that $\mathcal{I}_{j_B}$ form a set of zero modes localized by $j_B$ provided that all bosons with $j<j_B$ cannot condense. The proof then proceeds straightforwardly by induction. 

We apply our no-go theorem successively to bosons $j_B$ in range II in order of increasing $j_B$.
Using the result that all bosons in range I are uncondensable, one verifies that the particles $j$ with
$1\leq j\leq \mathrm{min}\{k-2j_B,\lfloor j_B/2\rfloor-1\}$
form a set $\mathcal{I}_{j_B}$.
As for range I, we can estimate the quantum dimensions. 
From the relation  
$\sin[\pi(2j_B+1)/(k+2)]=\sin[\pi(k-2j_B+1)/(k+2)]$ we can estimate the quantum dimension of $j_B$ using $\sin[\pi(2j_B+1)/(k+2)]<\pi(k-2j_B+1)/(k+2)$. The quantum dimensions of the anyons in  $\mathcal{I}_{j_B}$ are estimated as for range I with $\sin[\pi(2j+1)/(k+2)]<\pi(2j+1)/(k+2)$. Using these estimates we find that if
\begin{equation}
k-2j_B+1<\sum_{j=1}^{\mathrm{min}\{k-2j_B,\lfloor j_B/2\rfloor-1\}}(2j+1)
\label{eq: range II d estimate}
\end{equation}
 holds, Eq.~\eqref{eq: contradiction} follows. 
In the case $k-2j_B<\lfloor j_B/2\rfloor-1$, Eq.~\eqref{eq: range II d estimate} reduces to $1<(k-2 j_B)^2+(k-2 j_B)$, which is true for all $j_B$ in range II for all $k$. 
In the case $k-2j_B>\lfloor j_B/2\rfloor-1$, Eq.~\eqref{eq: range II d estimate} simplifies to $k+2<2j_B+(\lfloor j_B/2\rfloor)^2$, which holds for all $j_B$ in range II if $k\geq 37$. We verified explicitly that Eq.~\eqref{eq: contradiction} holds for all bosons in range II if $k<37$ (they appear in $k=13,19,31$). 
This concludes our proof that no condensation transition is possible in the SO(3)$_k$ TQFT for any odd $k$.

We note that this result can be readily extended to SU(2)$_k$ with $k$ odd, since SO(3)$_k$ is the projection of SU(2)$_k$ to anyons with integer $j$. One simply includes the half-integer $j$ anyons in the theory (none of which are bosons). The sets $\mathcal{I}_b$ as defined above remain the same and so do all the quantum dimensions. Hence, we also showed the noncondensability of SU(2)$_k$, with $k$ odd.
This is consistent with the ADE classification of SU(2)$_k$~\cite{ADE}: There are no off-diagonal modular invariant partition functions for odd $k$ in SU(2)$_k$~\cite{yellow}. Thus, the no-go theorem provides a proof of this fact that is complementary to the ADE classification.

{\it Example (iii): Multiple layers of} SO(3)$_k$ ---
We can show that any number of layers of SO(3)$_k$, with $k$ odd, does not contain condensable bosons. Fixing $k$, the proof proceeds again by induction. As induction base, we proof that all multi-layer anyons with a nontrivial particle in only a single layer (and the identity anyon in the other $k-1$ layers) cannot condense nor split. To show that, we can use the single-layer result from  Example (ii).
For the induction step, we assume that for a fixed $k_0<k$ all multi-layer anyons with nontrivial particles in $l$ layers, $1\leq l\leq k_0$, cannot condense and do not split. We can then show that the same holds for multilayer anyons with nontrivial particles in $k_0+1$ layers, completing the induction. The details of this proof are given in~\cite{Supp}.

{\it Summary} ---
We have presented a generally applicable no-go theorem against  the condensation of a topological boson and illustrated it with several examples. The proof of our theorem uses mostly the fusion (as compared to the braiding) information of the TQFT. 
We showed a connection between our results and the ADE classification of SU(2)$_k$ theories, indicating that the no-go theorem might be useful for the classification of modular invariant partition functions of conformal field theories more broadly.~\cite{Neupert16}
It would be interesting to study, whether other obstructions against boson condensation exist or whether our no-go theorem actually constitutes a necessary condition. In all examples we know, noncondensability is captured by the no-go theorem.

The no-go theorem can not only be used to directly classify (symmetry enriched) TQFTs, but also symmetry protected topological phases without intrinsic topological order. The  classification of the latter is often related to the former upon gauging the protecting symmetry.~\cite{Levin12,Chen13}

Finally, we note that a particle physics analogy of the no-go theorem is that a particle cannot decay if its mass is smaller than the total mass of its potential products. In this picture desintegration amounts to condensation and mass to quantum dimension.

\begin{acknowledgments}
The authors thank Parsa Bonderson for useful discussions.
This work was supported by NSF CAREER DMR-095242, ONR - N00014-11-1-0635, ARO MURI W911NF-12-1-0461, NSF-MRSEC DMR-1420541, FIS2012-33642, SEV-2012-0249, Packard Foundation and Keck grant. 
  \end{acknowledgments}

%----
\clearpage

\section{Supplemental Information}

%\maketitle

\subsection{Proof for Example (iii): Multiple layers of SO(3)$_k$}
\label{app: estimating quantum dimensions}

In this section, we show that no condensation is possible in the TQFT $\mathrm{SO}(3)_k^{\otimes N}$ comprised of $N$ layers of $\mathrm{SO}(3)_k$ for any odd $k$ and any integer $N$. 
The proof goes by induction. We denote the particles in  $\mathrm{SO}(3)_k^{\otimes N}$  with a shorthand notation. An anyon that has the identity particle from $\mathrm{SO}(3)_k$ in all layers, except for the $k_0$ layers $i_1,\ i_2,\ \cdots, i_{k_0}$, is denoted by  
$\{j_{i_1}j_{i_2}\cdots j_{i_{k_0}}\}$. Here $1\leq j_{i_l}\leq (k-1)/2$ can stand for any anyon from SO(3)$_k$ (except the identity 0), for all $l=1,\cdots, k_0$.

\subsubsection{Induction base}

First, consider particles $\{j_{i}\}$ with just one nontrivial anyon in some layer $i$. This will serve as the induction base.
By the no-go theorem and our proof in Example (ii), we know that no bosons of form $\{j_{i}\}$ can condense. [Use the particles with only one nontrivial in that same layer $i$ to build the set $\mathcal{I}_{j_i}$ as elaborated for Example (ii).] As a corollary, the anyons $\{j_{i}\}$ do not split: when fused with themselves no condensable boson appears in the fusion product, which prevents splitting by Eq.~(2) from the main text for $t=\varphi$.

\subsubsection{Induction step}

We assume that for any $1\leq l\leq k_0$ all $\{j_{i_1}j_{i_2}\cdots j_{i_{l}}\}$ 
\begin{enumerate}
\item do not condense and 
\item do not split.
\end{enumerate}
We now show the induction step, namely that all particles with nontrivial anyons in $(k_0+1)$ layers $\{j_{i_1}j_{i_2}\cdots j_{i_{k_0+1}}\}$ neither condense nor split.

We begin by showing that $\{j_{i_1}j_{i_2}\cdots j_{i_{k_0+1}}\}$ cannot condense.
The particles $\{j_{i_1}j_{i_2}\cdots j_{i_{k_0+1}}\}$ can be obtained by fusing a $\{j_{i_1}j_{i_2}\cdots j_{i_{k_0}}\}$ with a $\{j_{i_{k_0+1}}\}$, where $i_{k_0+1}\notin\{i_1,\cdots, i_{k_0}\}$. In this case, Eq.~(2) from the main text reads for $t=\varphi$
\begin{equation}
\tilde{N}^{\varphi}_{\{j_{i_1}\cdots j_{i_{k_0}}\}^{\downarrow},\{j_{i_{k_0+1}}\}^{\downarrow}}
=n_{\{j_{i_1}j_{i_2}\cdots j_{i_{k_0}}j_{i_{k_0+1}}\}}^{\varphi}.
\label{eq: n tilde Fib proof}
\end{equation}
%because $N^{\{j_{i_1}\cdots j_{i_{k_0}}j_{i_{k_0+1}}\}}_{\{j_{i_1}\cdots j_{i_{k_0}}\},\{j_{i_{k_0+1}}\}}=1$. 
Now, because of the uniqueness of the antiparticle, 
$\tilde{N}^{\varphi}_{\{j_{i_1}\cdots j_{i_{k_0}}\}^{\downarrow},\{j_{i_{k_0+1}}\}^{\downarrow}}
$ can be either 0 or 1. If it was 1, $\{j_{i_1}j_{i_2}\cdots j_{i_{k_0}}\}^{\downarrow}$ would be the antiparticle of $\{j_{i_{k_0+1}}\}^{\downarrow}$. Because all particles are their own antiparticles, this would imply  $\{j_{i_1}j_{i_2}\cdots j_{i_{k_0}}\}^{\downarrow}=\{j_{i_{k_0+1}}\}^{\downarrow}$. However, this is not possible for $k_0>1$, because the associativity of fusion would then also imply that $\{j_{i_1}\}^{\downarrow}$ is the antiparticle (and coinciding with) $\{j_{i_2}\cdots j_{i_{k_0}}j_{i_{k_0+1}}\}^{\downarrow}$, i.e.,
%\begin{equation}
$\{j_{i_1}\}^{\downarrow}=\{j_{i_2}\cdots j_{i_{k_0}}j_{i_{k_0+1}}\}^{\downarrow}$.
%\label{eq:associative fusion 2}
%\end{equation}
Remembering that $\{j_{i_1}j_{i_2}\cdots j_{i_{k_0}}\}$, $\{j_{i_2}j_{i_2}\cdots j_{i_{k_0}+1}\}$ do not split, and 
equating the quantum dimensions of the particles 
for these two identifications
%in Eq.~\eqref{eq:associative fusion 1} and Eq.~\eqref{eq:associative fusion 2} 
we have
\begin{equation}
\begin{split}
&d_{j_{i_1}}d_{j_{i_2}}\cdots d_{j_{i_{k_0}}}=d_{j_{i_{k_0+1}}},
\\
&d_{j_{i_1}}=d_{j_{i_2}}\cdots d_{j_{i_{k_0}}}d_{j_{i_{k_0+1}}}.
\label{eq: quantum dimension partitioning} 
\end{split}
\end{equation}
For $k_0>1$, this contradicts the fact that all nontrivial particles in this theory have quantum dimensions $d>1$. This rules out the possibility $\tilde{N}^{\varphi}_{\{j_{i_1}\cdots j_{i_{k_0}}\}^{\downarrow},\{j_{i_{k_0+1}}\}^{\downarrow}}=1$ and shows that $\{j_{i_1}j_{i_2}\cdots j_{i_{k_0+1}}\}$ does not condense for $k_0>1$.

The case $k_0=1$ needs to be considered separately, as both lines in Eq.~\eqref{eq: quantum dimension partitioning} are identical in this case, and therefore do not lead to a contradiction.
Assume that $\tilde{N}^{\varphi}_{\{j_{i_1}\}^{\downarrow},\{j_{i_{2}}\}^{\downarrow}}
=1$.
In the case $j_{i_1}\neq j_{i_2}$, we can rely on the following argument to disprove this assumption:  As all anyons in SO(3)$_k$ with $k$ odd have distinct quantum dimension, it follows that the two anyons $\{j_{i_1}\}$ and $\{j_{i_2}\}$ restrict to distinct particles and in particular $\varphi \notin j^{\downarrow}_{i_1}\times j^\downarrow_{i_2}$ -- with Eq.~(2) from the main text this implies that  $\{j_{i_1} j_{i_2}\}$ neither splits nor condenses. 
In the case $j_{i_1}=j_{i_2}\equiv j$, define $\hat{j}\equiv\{j_{i_1} j_{i_2}\}$.
We want to show that $\hat{j}$ does not condense. 
As there are no fermions in SO(3)$_k$ with $k$ odd, $\hat{j}$ can only be a boson if $\theta_j=1$, i.e., if $\{j_{i_1}\}$ and $\{j_{i_2}\}$ are bosons. 
Our no-go theorem applies to all bosons $\{j_{i_1}\}$ and $\{j_{i_2}\}$ with zero mode sets  $\mathcal{I}_{\{j_{i_1}\}}$ and $\mathcal{I}_{\{j_{i_2}\}}$.
We can then use the set $\mathcal{I}_{\hat{j}}=\mathcal{I}_{\{j_{i_1}\}}\times\mathcal{I}_{\{j_{i_2}\}}$, containing the fusion product of any particle in $\mathcal{I}_{\{j_{i_1}\}}$ with any particle in $\mathcal{I}_{\{j_{i_2}\}}$, to prove that $\hat{j}$ cannot condense.
To show that $\mathcal{I}_{\hat{j}}$ is a set of zero modes of $\hat{j}$, the main challenge is to show that the product of any two elements from  $\mathcal{I}_{\hat{j}}$ cannot condense. The product of any two elements from  $\mathcal{I}_{\hat{j}}$  is always of the form $\{j_{i_1} j_{i_2}\}$.  We have shown that when $j_{i_1}\neq j_{i_2}$ such particles cannot condense. We therefore need only show that non-trivial particles of form $\{j_{i_1} j_{i_2}\}$ with $j_{i_1}= j_{i_2}$ both bosons cannot condense. In order to show they are not condensable, we can use the proof given for Example (ii). For that, observe that the anyons $\hat{j}$ have the same fusion coefficients among themselves as the $j$ anyons in SO(3)$_k$ in Example~(ii) have, i.e., $N_{\hat{j},\hat{j}'}^{\hat{j}''}=N_{j,j'}^{j''}$, where $j,\ j', j'' \in $ SO(3)$_k$. Recall that conditions 1--3 from the definition of a set of zero modes only depend on the fusion coefficients $N_{j,j'}^{j''}$ and the information, which particles are bosons. Hence, conditions 1--3 are satisfied for  $\mathcal{I}_{\hat{j}}$ whenever they are satisfied for $\mathcal{I}_{j}$ in Example~(ii). It remains to show that $\mathcal{I}_{\hat{j}}$ is of large enough quantum dimension to satisfy the fundamental inequality Eq.~(4) from the main text.
For $\hat{j}$, Eq.~(4) from the main text takes the form
\begin{equation}
d_{\hat{j}}=d^2_j<
\Bigl(\sum_{a\in \mathcal{I}_{j}} d_a\Bigr)^2
=\sum_{\hat{a}\in \mathcal{I}_{\hat{j}}} d_{\hat{a}}.
\end{equation}
Upon taking the square root, this is equivalent to Eqs.~(14) and~(15) from the main text, which were shown to hold in Example (ii). 
Therefore the $\hat{j}=\{j_{i_1} j_{i_2}\}$ anyons do not condense and all $\{j_{i}\}$ have distinct restrictions.

We conclude that  for any $k_0\geq1$ only $\tilde{N}^{\varphi}_{\{j_{i_1}\cdots j_{i_{k_0}}\}^{\downarrow},\{j_{i_{k_0+1}}\}^{\downarrow}}=0$ is permitted and hence \eqnref{eq: n tilde Fib proof}
implies that $\{j_{i_1}j_{i_2}\cdots j_{i_{k_0+1}}\}$ does not restrict to the identity $\varphi$, i.e., it does not condense. This proves the assumption 1 of the induction step for $k_0+1$.

To complete the induction step, we need to show that $\{j_{i_1}j_{i_2}\cdots j_{i_{k_0+1}}\}$ does not split. For that, consider Eq.~(2) from the main text for $\{j_{i_1}\cdots j_{i_{k_0+1}}\}$ with itself and $t=\varphi$
\begin{align}
\sum_{r}\left(n^r_{\{j_{i_1}\cdots j_{i_{k_0+1}}\}}\right)^2=&\,
\sum_c N^c_{\{j_{i_1}\cdots j_{i_{k_0+1}}\},\{j_{i_1}\cdots j_{i_{k_0+1}}\}} \, n^{\varphi}_c\nonumber\\
=&\, n^{\varphi}_1=1.
\end{align}
We have used that  none of the $\{j_{i_1}j_{i_2}\cdots j_{i_{l}}\}$ with $1\leq l\leq k_0+1$ can restrict to the identity $\varphi$ since they cannot condense. This implies none of $\{j_{i_1}\cdots j_{i_{k_0+1}}\}$ splits, which proves the assumption 2 of the induction step for $k_0+1$.

We have thus shown inductively that none of the particles (except for the vacuum) restricts to the vacuum in the $N$-layer theory SO(3)$_k^{\otimes N}$. Thus, there is no condensate and with it no condensation in any number $N$ of layers of SO(3)$_k$ with $k$ odd.

\subsection{No-go theorem with Abelian sector}
\label{app: abelian no-go}

We have seen from the examples discussed in the main text, that the no-go theorem can often be used to not only show that individual bosons in a TQFT cannot condense, but that an entire TQFT is not condensable. Here, we extend this discussion to examples of TQFTs that have noncondensable sub-structures.
This problem is motived by physical examples: in the fractional quantum Hall effect, for instance, one frequently discusses phases that are described by a direct (or semi-direct) product of an Abelian and a non-Abelian TQFT. A simple example is the $\mathbb{Z}_3$ Read-Rezayi state of bosons, which is described by the TQFT $\mathcal{A}_{\mathrm{Fib}}\times\mathbb{Z}_2$. While such a theory admits condensations, already in the $\mathbb{Z}_2^{\otimes N}$ sector,  when enough layers are considered, one has the intuition that the noncondensability of Fibonacci should still constrain the possible condensations.

\begin{mylem}
Consider a TQFT
\begin{equation}
\mathcal{A}\times\mathcal{X},
\end{equation}
where $\mathcal{X}$ is an Abelian TQFT (i.e., all its anyons have quantum dimension 1). Further,  
for all particles $b\in\mathcal{A}$ (not only for the bosons), except for the vacuum, let there exist a set $\mathcal{I}_b=\{a_1,\ldots,a_m\}$  of zero modes of $b$, containing anyons from $\mathcal{A}$, such that
 the quantum dimensions satisfy
\begin{equation}
d_b\leq\sum_{i=1}^m d_{a_i}.
\label{eq: quantum dimensional obstruction for any particle}
\end{equation}
Then, any possible condensation transition will lead to a theory of the form
\begin{equation}
\mathcal{A}\times\mathcal{Y},
\label{eq: final product theory}
\end{equation}
where the Abelian TQFT $\mathcal{Y}$ can be obtained from $\mathcal{X}$ through a condensation. 
\label{abelian Lemma}
\end{mylem}

\begin{proof}
This Lemma follows almost directly from the no-go theorem. Let us denote a particle from $\mathcal{A}\times\mathcal{X}$ by the pair $(b,x)$ where $b\in\mathcal{A}$ and $x\in \mathcal{X}$.
If $(b,x)$ is boson, we can show that it has to be an uncondensable one, except if $b=1$. The set 
\begin{equation}
\mathcal{I}_{(b,x)}=\left\{
(a_1,x),\cdots, (a_m,x)
\right\},
\end{equation}
(where $a_1,\cdots, a_m$ form a set $\mathcal{I}_b$ of zero modes of $b$ whose existence is guaranteed by assumption) satisfies all the conditions 1--3 form the definition of a set of $(b,x)$ zero modes. Since $x$ is an Abelian particle, $d_x=1$ and Eq.~\eqref{eq: quantum dimensional obstruction for any particle} directly implies that the sum of the quantum dimensions of the particles in $\mathcal{I}_{(b,x)}$ satisfies the inequality~(4) from the main text. Hence, $(b,x)$ cannot condense. 
In turn, this implies any condensable boson in $\mathcal{A}\times\mathcal{X}$ is of the form $(1,x)$. A condensate of this form is transparent to the anyons in $\mathcal{A}$ and will thus leave this sub-TQFT unaffected. It will only induce a condensation $\mathcal{X}\to\mathcal{Y}$, so that the final theory is of the from~\eqref{eq: final product theory}.
\end{proof}

We return to the example of $\mathcal{A}_{\mathrm{Fib}}\times\mathbb{Z}_2$. Consider $N$ layers of this theory, i.e., $\mathcal{A}_{\mathrm{Fib}}^{\otimes N}\times\mathbb{Z}_2^{\otimes N}$. This multi-layer TQFT satisfies all assumptions of Lemma~\ref{abelian Lemma}: for each anyon $b\in \mathcal{A}_{\mathrm{Fib}}^{\otimes N}$, a choice for the set $\mathcal{I}_b$ is given by $\mathcal{I}_b=\{1,b\}$. This is so because all possible bosons appearing in the fusion product of $b\times b$ are uncondensable by the no-go theorem and the sum of the quantum dimensions of $\mathcal{I}_b$, given by $1+d_b$ is larger than $d_b$. We conclude that the $\mathcal{A}_{\mathrm{Fib}}^{\otimes N}$ structure is preserved under any condensation transition in such a theory.

\subsection{General constraints on boson condensation}
\label{sec: general lemmas}

In this section, we list lemmas that pose other general constraints on condensation transitions in TQFTs.

\begin{mylem}
Suppose $S=\left\{ a_{1},\cdots, a_{m}\right\} $ is a collection of particles in a TQFT $\mathcal{A}$
with $a_{i}\times\bar{a}_{i}$ containing no bosons other than the identity
-- i.e., $n_{a_{i}}^{t}=\delta_{a_{i}^{\downarrow}}^{t}$ and $a_{i}$
does not split. Moreover assume $a_{i}^{\downarrow}\neq a_{j}^{\downarrow}$ for $i\neq j$. Then if a boson $B$ appears in the fusion of $a_{i}$ and $\bar{a}_{j}$,  $a_{i}\times \bar{a}_{j}=B+\cdots$ for any $i\neq j$, then
$B$ is not condensable. \end{mylem}
\begin{proof}
Using Eq.~(2) from the main text for $a=a_{i}$, $b=\bar{a}_{j}$ and $t=\varphi$, we have $\delta_{ij}=\sum_{t}n_{a_{i}}^{t}n_{\bar{a}_{j}}^{t}=\sum_{c}n_{c}^{\varphi}N_{a_{i}\bar{a}_{j}}^{c}$.
For $i\neq j$ we get \textbf{$\sum_{B}n_{B}^{\varphi}N_{a_{i}\bar{a}_j}^{B}=0$}.
So if boson $B$ appears in $a_{i}\times \bar{a}_{j}$, we must have $n_{B}^{\varphi}=0$,
so that $B$ is not condensable. \end{proof}

\begin{mylem}
Consider a TQFT $\mathcal{A}$ with no fusion multiplicity and just one boson $B$. If $B$ is condensed then either $B$ is abelian or $B^{\downarrow} = \varphi + r $ where $r$ is a single anyon. 
\label{lem: single boson lemma}
\end{mylem}
\begin{proof}
As there is just a single boson, $B=\overline{B}$. Equation~(2) from the main text implies 
\begin{equation}
\sum_{t} n^{t}_{B} n^{t}_{B} = \sum_{c} n^{\varphi}_{c} N^{c}_{B B}= 1 +n^{\varphi}_{B} N^{B}_{B B}.
\end{equation}
Notice, however, that the left-hand side is greater or equal to $n^{\varphi}_{B} n^{\varphi}_{B}$. For condensation, this implies $n^{\varphi}_{B}=1$, and tells us that $\sum_{t} n^{t}_{B} n^{t}_{B} =1 \text{ or } 2$. In the former case, $B^{\downarrow} = \varphi$. This implies $d_{B} = 1$, and so $B$ is a quantum dimension $1$ boson hence must have $N^{a}_{BB}=\delta_{a,1}$. In the latter case, $B$ restricts to just two particles with multiplicity $1$ each, so that $B^{\downarrow}=\varphi+r$. 
\end{proof}

\begin{mylem}
With the  conditions of Lemma~\ref{lem: single boson lemma}, and assuming $B$ has $d_{B}>1$, condensation of $B$ can only occur if $N^{a}_{BB} N^{b}_{BB} \leq N^{B}_{ab}$ for all anyons $a$ and $b$ of $\mathcal{A}$.
 \end{mylem}
\begin{proof}
Lemma~\ref{lem: single boson lemma} shows $B^{\downarrow}=\varphi+r$, where $r$ a simple object. Consider $a \in \mathcal{A}$ where $a\neq 1,B$. Equation~(2) from the main text  for $b=B$ and $t=\varphi$ reads 
\begin{equation}
\begin{split}
n^{r}_{a} n^{r}_{B}
=& n^{r}_{a}  \\
=& N^{B}_{a B} n^{\varphi}_{B}\\
=& N^{B}_{a B}
\end{split}
\end{equation}
Consider now Eq.~(2) from the main text for $b\neq 1,B$ and for $a\neq b$, which gives
\begin{equation}
\sum_{t} n^{t}_{a} n^{t}_{b}= N^{B}_{ab} \geq  n^{r}_{a} n^{r}_{b}
\end{equation}
Combining the $a,B$ and $b,B$ and $a,b$ equations gives the inequality
\begin{equation}
 N^{b}_{B B} N^{a}_{B B}  \leq N^{B}_{ab}.
\end{equation}

\end{proof}

\clearpage
\newpage
%\bibliography{bib}

\begin{thebibliography}{99}
\bibitem{footnote-Ib-set}
In demanding that $a_i\times a_j$ does not contain condensable bosons, as opposed to not containing any bosons at all (except the identity), we are anticipating a inductive application of the no-go theorem. Once we have shown that a boson $B$, whose set  $\mathcal{I}_B$ is such that  $a_i\times a_j$, with $a_i, a_j\in \mathcal{I}_B$, does not contain any boson (except the identity), 
is uncondensable, it is allowed that $B$ appears in the fusion product $a_i\times a_j$ of the set $\mathcal{I}_{B'}$ of another boson $B'$.
\end{thebibliography}

\end{document}